%% file: SimplerProof.tex
\documentclass[10pt,a4paper]{article}
\usepackage{a4wide}
\usepackage{amsmath}
\usepackage{html}
\usepackage{url}
\usepackage{epsfig}
\usepackage{psfrag}
\usepackage{amssymb}
\usepackage{wrapfig}
\usepackage{algorithm}
\usepackage{algorithmic}
\usepackage{paralist}
\usepackage{xcolor}

\input papermacros

\usepackage{mathptmx}
\usepackage[scaled]{helvet}   
\usepackage{courier}

\usepackage[T1]{fontenc}
\usepackage[utf8]{inputenc}

\usepackage{xfrac}
\usepackage{thmtools}
\usepackage{amsmath}
\usepackage{amssymb}

\newcommand{\NSW}{\mathit{NSW}}
\newcommand{\Ahalf}{A_{\sfrac{1}{2}}}\newcommand{\Ohalf}{O_{\sfrac{1}{2}}}

\newcommand{\half}{{\sfrac{1}{2}}}
\newcommand{\mhalf}{m_{\half}}
\newcommand{\nhalf}{n_{\half}}

\renewcommand{\sset}[1]{\{ #1 \}}

\newcommand{\hdeg}{\mathit{hdeg}}

\newcommand{\NN}{\mathbf{N}}

\newcommand{\AH}{A^H}
\newcommand{\OH}{O^H}

\begin{document}
\title{Maximizing Nash Social Welfare in 2-Value Instances: A Simpler Proof for the Half-Integer Case}
\author{Kurt Mehlhorn}

\maketitle

\begin{abstract} A set of $m$ indivisible goods is to be allocated to a set of $n$ agents. Each agent $i$ has an additive valuation function $v_i$ over goods. The value of a good $g$ for agent $i$ is either $1$ or $s$, where $s$ is a fixed rational number greater than one, and the value of a bundle of goods is the sum of the values of the goods in the bundle. An \emph{allocation} $X$ is a partition of the goods into bundles $X_1$, \ldots, $X_n$, one for each agent. The \emph{Nash Social Welfare} ($\NSW$) of an allocation $X$ is defined as
\[ \NSW(X) = \left( \prod_i v_i(X_i) \right)^{\sfrac{1}{n}}.\]
The \emph{$\NSW$-allocation} maximizes the Nash Social Welfare. In~\cite{NSW-twovalues-halfinteger} it was shown that the $\NSW$-allocation can be computed in polynomial time, if $s$ is an integer or a half-integer, and that the problem is NP-complete otherwise. The proof for the half-integer case is quite involved. In this note we give a simpler and shorter proof. \end{abstract}

\section{Introduction}

A set of $m$ indivisible goods is to be allocated to a set of $n$ agents. Each agent $i$ has an additive valuation function $v_i$ over goods. The value of a good $g$ for agent $i$ is either $1$ or $s$, i.e., $v_i(g) \in \sset{1,s}$; here $s$ is a fixed rational number greater than one. 
The value of a bundle $X_i$ of goods for $i$ is the sum of the values of the goods in the bundle, i.e., $v_i(X_i) = \sum_{g \in X_i} v_i(g)$. We call a good $g$ \emph{heavy} if $v_i(g) = s$ for some $i$, and \emph{light} otherwise, i.e., $v_i(g) = 1$ for all $i$. An \emph{allocation} $X$ is a partition of the goods into bundles $X_1$, \ldots, $X_n$, one for each agent, i.e., $\cup_i X_i = [m]$ and $X_i \cap X_j = \emptyset$ for $i \not= j$. The \emph{Nash Social Welfare} (NSW) of an allocation $X$ is defined as
\[ \NSW(X) = \left( \prod_i v_i(X_i) \right)^{\sfrac{1}{n}}.\]
The \emph{$\NSW$-allocation} maximizes the Nash Social Welfare. In~\cite{NSW-twovalues-halfinteger} it was shown that the $\NSW$-allocation can be computed in polynomial time, if $s$ is an integer or a half-integer, and that the problem is NP-complete otherwise. The proof for the half-integer case is quite involved. In this note we give a simpler and shorter proof.

In~\cite{NSW-twovalues-halfinteger}, it was argued that for the case of $s$ being integral, the allocation of the heavy and the light goods can be considered separately, and that this is not the case when $s$ is half-integral. The following example is given. Consider two agents with identical valuations and $s = \sfrac{3}{2}$. There are two heavy goods and either two or three light goods. If there are two light goods, each agent should receive a heavy and a light good, and if there are three light goods, one agent should receive the two heavy goods and the other should receive the three light goods. So, the number of available light goods influences the allocation of the heavy goods. The authors conclude that it is necessary to consider heavy and light goods together throughout the algorithm. We give an algorithm that consists of two phases and considers the light goods only in the second phase. This is the first simplification.

The proof in~\cite{NSW-twovalues-halfinteger} is algorithmic, i.e., the paper introduces a collection of improvement rules and shows that their repeated application produces a $\NSW$-allocation. In contrast, we characterize the $\NSW$-allocation and discuss the algorithmics after the characterization. This is the second simplification.

Whenever an improvement rule is applied in~\cite{NSW-twovalues-halfinteger}, a detailed accounting of how the heavy and light goods are reassigned is given. In many situations we are able to replace the detailed accounting by a global argument. This is the third simplication. 

It is convenient to represent an instance of our problem as a bipartite graph. The two sides of the graph are the agents and the goods respectively, each agent is connected to each good, and the edge from agent $i$ to good $g$ is labelled as either heavy or light. An allocation $A$ is an assignment of the goods to the agents. \emph{In this paper, we restrict attention to allocations that assign every heavy good to an agent that considers it heavy}; the general case is treated in Section 4.4 of~\cite{NSW-twovalues-halfinteger} and we have no simplication to offer for this part. The restricted case is handled in Sections 4.1, 4.2, and 4.3 of~\cite{NSW-twovalues-halfinteger} that span a total of 25 pages. The current presentation has only half the length. 

\paragraph{Related Work:} In general, determining the optimal $\NSW$ is NP-complete. Good approximation algorithms were obtained by~\cite{Cole-Gkatzelis,CDGJMVY17,AnariGSS17,DBLP:conf/sigecom/BarmanKV18}. The current best factor is $e^{\sfrac{1}{e}} \approx 1,445$. For binary valuations, i.e., $v_i(g) \in \sset{0,1}$ for all $i$ and $g$, the $\NSW$-allocation can be determined in polynomial time~\cite{DBLP:conf/atal/BarmanKV18}.
For two-valued valuations, i.e., $v_i(g) \in \sset{1,s}$ for all $i$ and $g$, and $s$ a rational number greater than one, the $\NSW$-allocation can be computed in polynomial time if $s$ is an integer or a half-integer. Otherwise, the problem is NP-complete~\cite{NSW-twovalues-halfinteger}.

\section{The Algorithm}

The algorithm consists of two steps.
\begin{enumerate}
\item 
Determine the lexmin allocation of the heavy goods, i.e., push heavy goods towards smaller bundles as much as possible. Let $b_n \ge b_{n-1} \ge \ldots \ge b_1$ be the number of heavy goods allocated to the different agents sorted in decreasing order. Then $b_n$ is minimal among all allocations of the heavy goods, and given that $b_n$ has its minimal value, $b_{n-1}$ is minimal, and so on. In other words, the string $b_nb_{n-1}\ldots b_1$ is lexicographically minimal. 

\item Allocate the light goods greedily, i.e., allocate the light goods one by one and always add the next good to a bundle of smallest value. Let $x$ be the minimum value of any bundle in the resulting allocation. Call the bundles of value $x$, $x + \half$, and $x + 1$ the small bundles and let $N_s$ be the owners of the small bundles. Optimize the allocation of the small bundles, i.e., allocate the goods contained in the small bundles to the agents in $N_s$ so as to allocate heavy goods only to agents considering them heavy, give each bundle a value in $[x,x+1]$, and maximize the number of bundles of value $x + \half$.
\end{enumerate}
In the next subsections, we give more details.

\subsection{The LexMin Allocation of the Heavy Goods}

We need the concept of an alternating path. Let $A$ be an allocation, let $A^H$ be its restriction to the heavy edges, and let $E^H$ be the set of all heavy edges.  An alternating path (with respect to $A^H$ and $\bar{A} = E^H \setminus A^H$) uses alternatingly edges\footnote{We will later also consider alternating path with respect to $A$ and $O$ that use alternatingly edges in $A^H$ and $\OH$.} in $A^H$ and $\bar{A}$. The first and the last edge of the path have different types; one belongs to $A$ and one belongs to $\bar{A}$. The endpoint incident to the $A$-edge is called the $A$-endpoint of the path. The other endpoint is then the $\bar{A}$-endpoint. Augmentation of an alternating path to $A^H$ effectively moves a heavy good from the $A$-endpoint to the $\bar{A}$-endpoint.  

The lexmin allocation of the heavy goods is readily computed. Start with an arbitrary allocation $A^H$ of the heavy goods. As long as there is an alternating path starting with an $A$-edge from an agent $i$ to an agent $j$ owning at least two heavy goods less than $i$, augment the path. This decreases the number of heavy goods at $i$ by one and increases the number of heavy goods at $j$ by one.\footnote{Alternatively, one can use parametric flow~(\cite{ParametricFlow, Darwish-Mehlhorn}). Have an edge of capacity $s$ from an agent to each good that the agent considers heavy, and add a source and a sink. Have an edge of capacity $\lambda$ from the source to each agent, and an edge of capacity $s$ from each good to the sink. Determine max-flows for all values of $\lambda$. For $\lambda \le b_1$ the flow is $\lambda n$. For $b_1 \le \lambda \le b_2$ the max-flow is
$b_1 + \lambda(n-1)$, and so on.}

Let $\bar{k}$ (called $b_n$ above) be the maximum number of goods assigned to any agent. Define $R'_{\bar{k} + 1} = \emptyset$ and sets $R_\ell$ and $R'_\ell$ of agents for $\ell = \bar{k}$, $\bar{k} -1$, $\bar{k} - 2$, and so on. Let
\begin{align*}
  R_\ell & = \set{a}{\text{$a$ owns $\ell$ goods and does not belong to $R'_{\ell + 1}$}},\\
           R'_\ell &= \set{a}{\text{$a$ owns $\ell - 1$ goods and there is an alternating path starting with an $A$-edge from an agent in $R_\ell$ to $a$}}.
\end{align*}

\begin{lemma} For any $\ell$, agents in $R_\ell \cup R'_\ell$ own either $\ell$ or $\ell - 1$ heavy goods. Each of them could own $\ell$ goods via a transfer from another agent in the set. For any $k$, all heavy goods assigned to the agents in $\cup_{\ell \ge k} R_\ell \cup R'_\ell$ must be assigned to them.\end{lemma}
\begin{proof} By construction, there is no alternating path starting with an $A$-edge from an agent in $\cup_{\ell \ge k} R_\ell \cup R'_\ell$ to an agent in $\cup_{\ell < k} R_\ell \cup R'_\ell$, and each agent in $R'_k$ can be reached from an agent in $R_k$ by such an alternating path. \end{proof}

We use $H$ to denote the lexmin allocation of the heavy goods.


\section{Assignment of the Light Goods and Optimization of the Small Bundles}\label{APC exists}

We add the light goods greedily to the lexmin allocation of the heavy goods, i.e., we add the light goods one by one and always to a bundle of minimum weight. Let $x$ be the minimum weight of a bundle after the addition of the light goods. Then only bundles of weight $x$, $x + \half$ and $x + 1$ contain light goods. The assignment of the light goods is somewhat arbitrary in the following sense: If a bundle of value $x + 1$ contains a light good, we may move this light good to a bundle of value $x$, thus interchanging the values of the bundles. 


Let $k_0$ be minimal such that $k_0 s > x + 1$. Then $(k_0 - 1) s \le x + 1$ and hence light goods may be added to the bundles in $R'_{k_0}$. At most $\floor{s}$ light goods are added to a bundle in $R'_{k_0}$; $\floor{s}$ goods are added if $k_0 s = x + \sfrac{3}{2}$ and the value of the bundle is raised to $x + 1$ by the addition of light goods.

At this point, we have bundles of value $x$, $x + \half$, $x + 1$, $k_0s$, $(k_0 + 1)s$, \ldots. We use $B$ to denote the allocation after the greedy addition of the light goods. Let us call the bundles of value at most $x + 1$ the \emph{small bundles}, the bundles of value $k_0s$ and more the \emph{large bundles}, let us write $N_s$ for the owners of the small bundles, and let $\Gamma(N_s)$ be the goods owned by them. The agents in $R'_{k_0}$ belong to $N_s$.

No allocation can assign more heavy goods to the agents in $N_s$ as $B$ does. Also $B$ allocates all light goods to $N_s$. So no allocation can assign more value to the agents in $N_s$ as $B$ does. If all bundles of $B$ in $N_s$ have value $x$, $B$ is clearly optimal. So we may assume that the average value of the bundles in $N_s$ lies strictly between $x$ and $x + 1$. 

We optimize the allocation of the small bundles by reassigning the goods contained in them so as to maximize the number of bundles of value $x + \half$, i.e., we consider an allocation of the goods in $\Gamma(N_s)$ to the agents in $N_s$ such that
\begin{enumerate}[(1)]
\item each heavy good in $\Gamma(N_s)$ is allocated to an agent in $N_s$ considering it heavy,
\item all bundles have value $x$, $x + \half$, or $x + 1$, and 
\item the number of bundles of value $x + \half$ is maximum;
\item subject to that the number of bundles of value $x + \half$ is maximum, the number of not-heavy-only bundles of value $x + 1$ is maximum.  
\end{enumerate}
We use $A$ to denote the allocation obtained after optimization of the small bundles. It consists of $H$ restricted to the bundles of value $k_0s$ and more plus the optimized allocation, call it $C$, to the agents in $N_s$. \emph{We will show that $A$ is optimal}.

Let $n_0$, $\nhalf$, and $n_1$ be the number of bundles in $C$ of value $x$, $x + \half$, and $x + 1$, respectively, let $S$ be their total value and let $n_s = \abs{N_s}$.  Then $n_s = n_0 + \nhalf + n_1$ and $S = n_0x + \nhalf (x + \half) + n_1 (x + 1)$ or $S - n_sx= \nhalf/2  + n_1$. Thus $n_1 = (S - n_sx) - \nhalf/2$ and $n_0 = n_s - \nhalf - n_1 = (n_s(x + 1) - S) - \nhalf/2$. The $\NSW$ of $C$ is therefore
\begin{equation}\label{large mhalf is good} x^{n_0} (x + \half)^{\nhalf} (x + 1)^{n_1} = x^{n_s(x+1) - S}(x + 1)^{S - n_sx} [(x + \half)^2/(x(x+1))]^{\nhalf/2}.
  \end{equation}
  Maximizing the number of bundles of value $x + \half$ in $C$ is therefore tantamount to maximizing the $\NSW$ of $C$. It is however not clear at this point that an optimal allocation must satisfy condition (2).

An efficient algorithm for constructing $C$ will be discussed in Section~\ref{Algorithmics}. Let $O$ be an optimal allocation closest to $A$, i.e., with minimal cardinality of $A^H \oplus O^H$, and let $x_O$ be the minimum value of a bundle in $O$. We decompose $A^H \oplus O^H$ into alternating paths. A good has degree zero or two in $A^H \oplus O^H$. In the latter case, one path goes through the good. In an agent $v$, we pair $A$- and $O$-edges in $A^H \oplus O^H$ as much as possible and leave the remaining $A$- or $O$-edges as starting edges of paths. Let $\hdeg_O(i)$ and $\hdeg_A(i)$ be the number of heavy goods owned by $i$ in $O$ and $A$, respectively. Let $w_i^O$ and $w_i^A$ be the value of the bundles owned by $i$ in $O$ and $A$, respectively. The \emph{utility profile} of an allocation is the multiset of the values of its bundles. 

\begin{lemma}\label{x0} $x_O \le x + \sfrac{1}{2} \le k_0s - 1$. Bundles of value more than $k_0s$ do not contain a light good in $O$.  \end{lemma}
\begin{proof} Since $O$ cannot assign more value to the agents in $N_s$ than $A$ and since in $A$ the average value of a bundle in $N_s$ is less than $x + 1$, $x_O < x + 1$.

If a bundle of value more than $k_0s$ would contain a light good, moving the light good to a bundle of value $x_O$ would improve the $\NSW$ of $O$. \end{proof}

 We will first show that $O$ and $A$ agree on the heavy bundles. More precisely, 

\begin{lemma}\label{agree on heavy} $A^H$ and $O^H$ agree on $\cup_{\ell \ge k_0} R_\ell \cup R'_\ell$. Bundles in $N_s$ contain at most $k_0 - 1$ heavy goods in $O$. \end{lemma}
\begin{proof}
We start with the first claim. Assume otherwise and let $k \ge k_0$ be maximal such that $A^H$ and $O^H$ do not agree on the agents in $R_k \cup R'_k$. Since $A^H$ and $O^H$ agree on all agents in $\bigcup_{h > k} R_h \cup R'_h$ and since the heavy goods assigned by $A^H$ to the agents in $\bigcup_{h \ge k} R_h \cup R'_h$ must be assigned to these agents in any allocation, the heavy goods assigned to the agents in $R_k \cup R'_k$ by $A$ must also be assigned to them by $O$. Since $\OH$ is closest to $\AH$, we even have $\AH_i \subseteq \OH_i$ for all $i \in R_k \cup R'_k$.  So there must be an $i \in R_k \cup R'_k$ which owns more heavy goods in $O$ than in $A$, i.e., $\hdeg_O(i) > \hdeg_A(i) \ge k-1$ and hence the heavy value of $i$ in $O$ is at least $k_0s$. Thus $O_i$ contains no light good as otherwise the value of $O_i$ would be larger than $k_0s$, a contradiction to Lemma~\ref{x0}. Let $P$ be a maximal alternating path with endpoint $i$ and let $j$ be its other endpoint. Then $\hdeg_O(j) < \hdeg_A(j)$. Let $\ell$ be the number of light goods owned by $j$. 

  If $\hdeg_O(j) = \hdeg_O(i) - 1$, augmenting the path to $O$ and also interchanging the light goods owned by $i$ and $j$ does not change the utility profile and hence the $\NSW$ of $O$ and moves $O$ closer to $A$, a contradiction.

  If $\hdeg_O(j) \le \hdeg_O(i) - 2$ and $w_j^O < w_i^O$, we augment the path and  move $\min(\ell,\floor{s})$ light goods from $j$ to $i$. This moves $O$ closer to $A$ and does not decrease the $\NSW$ of $O$ (it improves it if $w_j^O \le w_i^O - 1$). If $w_j^O \ge w_i^O$, $j$ owns at least $2s$ light goods.  Since $w^O_i \ge k_0s$ and bundles of value larger than $k_0s$ contain no light good, we have $w_j^O = w_i^O = k_0s$ and $k_0s \le x_O + 1$. Thus $k_0s = x_O + 1$. We move a light good from $j$ to a bundle of value $x_O$, augment the path and move $\floor{s}$ light goods from $j$ to $i$. This improves the $\NSW$ of $O$, a contradiction. 

If $\hdeg_O(j) \ge \hdeg_O(i)$, then $\hdeg_A(j) > \hdeg_O(j) \ge \hdeg_O(i) > \hdeg_A(i) \ge k - 1$. Thus
$\hdeg_A(j) \ge \hdeg_A(i) + 1$ and there is an alternating path starting with an $A$-edge from $j$ to $i$. Thus $A$ is not lexmin, a contradiction.
This completes the proof of the first claim.

For the second claim, assume that there is an agent $i \in N_s$ owning $k_0$ or more heavy goods in $O$. Since $i$ owns at most $k_0 - 1$ heavy goods in $A$, we have $\hdeg_O(i) \ge k_0 > \hdeg_A(i)$. As above, we conclude that $O_i$ contains no light good and consider an alternating path $P$ with endpoint $i$. Let $j$ be the other endpoint. Then $\hdeg_O(j) < \hdeg_A(i)$.

If $\hdeg_O(j) < hdeg_O(i)$, we argue as above. If $\hdeg_O(j) \ge \hdeg_O(i) \ge k_0$, then $\hdeg_A(j) > \hdeg_O(j) \ge \hdeg_O(i) > \hdeg_A(i)$, $P$ is an alternating path starting with an $A$-edge from $j$ to $i$ and $A$ is not lexmin.
This completes the proof of the second claim.
\end{proof}

\begin{lemma} $A$ and $O$ agree on the bundles $R_{k_0} \cup \bigcup_{h > k_0} R_h \cup R'_h$, i.e., on the agents owning $k_0$ or more heavy goods in $A$. \end{lemma}
\begin{proof} By the preceeding Lemma, they agree on the heavy goods. Also none of the bundles contains a light good in either $A$ or $O$.  \end{proof}

We still need to show that $A$ is the optimal allocation for the agents in $N_s$. \emph{Since we know at this point that $A$ and $O$ agree on the large bundles, we use $A$ and $O$ in the sequel for the bundles of $A$ and $O$ allocated to the agents in $N_s$}.

\begin{lemma}\label{less-more} If $A$ is suboptimal, there are more bundles of value $x + \sfrac{1}{2}$ in $O$ than in $A$ and less bundles of value $x + 1$.  \end{lemma}
\begin{proof} Let $m_d$ and $n_d$ be the number of bundles of value $x + d$ in $A$ and $O$, respectively. 
Only $m_0$, $\mhalf$, and $m_1$ are non-zero. There might be non-zero $n_k$ for $k < 0$ and $k > 1$. For the sake of a contradiction, assume $\nhalf \le \mhalf$. We massage $O$ by moving portions of $\half$ around; we maintain the number of bundles and their total value, but do not insist that $O$ stays a feasible allocation. Each move will strictly increase the $\NSW$, and we will end up with $A$. Thus $O = A$.

We observe first that the sums $\sum_{d \ge \sfrac{1}{2}} n_d$ and $\sum_{d \le \sfrac{1}{2}} n_d$ must both be non-zero, as otherwise the average weight of a bundle in $O$ would be at most $x$ or at least $x + 1$.  

If $n_d = 0$ for $d < 0$ and $d > 1$, $\nhalf \le \mhalf$ implies $\NSW(O) \le \NSW(A)$ by equation (\ref{large mhalf is good}). If $n_d > 0$ for some $d < 0$ or $d > 1$, one of the following rules improves the $\NSW$ of $O$ and maintains $\nhalf \le \mhalf$.

If $e + 1 \le d$ and $n_e > 0$ and $n_d > 0$, decrease $n_e$ and $n_d$ by one each and increase $n_{e + \sfrac{1}{2}}$ and $n_{d + \sfrac{1}{2}}$ by one each except if this would increase $\nhalf$ above $\mhalf$. If this rule is not applicable, we either have $n_d = 0$ for $d < 0$ and $d > 1$ and are done, or $\nhalf = 0 = \mhalf$ and $n_d = 0$ for $d \le -\sfrac{1}{2}$ or $d \ge \sfrac{3}{2}$ and either $n_{-\sfrac{1}{2}}$ or $n_{\sfrac{3}{2}}$ is non-zero (but not both). 
Since the total value of the bundles in $A$ and $O$ must be same, we have $-n_{-\sfrac{1}{2}}/2 + 0n_0 + n_1 + 3n_{\sfrac{3}{2}}/2 = 0m_0 + m_1$. So $n_{-\sfrac{1}{2}}$ and $n_{\sfrac{3}{2}}$ are even; note that only one can be non-zero.

If $n_{-\sfrac{1}{2}}$ is positive, we decrease it by two, decrease $n_1$ by one (it must be positive), and increase $n_0$ by three. This leaves the number of bundles and their total value unchanged. It increases $\NSW$ as
    \[ \frac{x^3}{(x - \half)^2 (x + 1)} = \frac{x^3}{(x^2 - x + \sfrac{1}{4})(x + 1)} = \frac{x^3}{x^3 - \sfrac{3x}{4} + \sfrac{1}{4}} > 1\]
    for $x \ge \sfrac{1}{2}$. 

    If $n_{\sfrac{3}{2}}$ is positive, we decrease it by two, decrease $n_0$ by one (it must be positive), and increase $n_1$ by three. This leaves the number of bundles and their total value unchanged. It increases $\NSW$ as
    \[ \frac{(x + 1)^3}{(x + \sfrac{3}{2})^2 x} = \frac{x^3 + 3x^2 + 3x + 1}{x^3 + 3x^2 + \sfrac{9}{4}x} > 1.\]

 A similar argument shows that there are less bundles of value $x + 1$ in $O$ than in $A$. Assume, for the sake of a contradiction, $n_1 \ge m_1$. If $n_e > 0$ and $n_d > 0$ for some $e < 1$ and $d > 1$, we decrease both counts by one and increase $n_{e + \sfrac{1}{2}}$ and $n_{d + \sfrac{1}{2}}$ by one each. So, we may assume that either $n_e = 0$ for all $e < 1$ or $n_d = 0$ for all $d > 1$. The former is impossible, as then the average weight of a bundle would be larger than $x + 1$. So $\sum_{e < 1} n_e > 0$. If $n_e > 0$ for some $e < 0$, we decrease $n_e$ and $n_{\sfrac{1}{2}}$ by one (the latter is positive by part one) and increase $n_{e + \sfrac{1}{2}}$ and $n_0$ by one. In this way, all bundles have values in $\sset{x, x + \sfrac{1}{2}, x + 1}$. Since the number of bundles and their total value is fixed, having more bundles of value $x + \sfrac{1}{2}$ implies having less bundles of value $x$ and $x + 1$. 
  \end{proof}

\newcommand{\rhalf}{r_{\sfrac{1}{2}}}
\begin{lemma}[Lemma 26 from~\cite{NSW-twovalues-halfinteger}]\label{max number of heavy} Let $r_d$ be the maximum number of heavy goods that a bundle of value $x + d$ may contain. Then $r_d = \max\set{r \in \NN_0}{x + d -sr \in \NN_0}$ and
  \begin{itemize}
  \item $r_0 = \rhalf - 1$ and $r_1 = \rhalf + 1$ if $x + 1$ is an integer multiple of $s$.
  \item $r_1 = r_0 = \rhalf + 1$ if $x + 1$ is not an integer multiple of $s$ and $x + 1 > (\rhalf + 1) s$.
  \item $r_1 = r_0 = \rhalf - 1$ if $x + 1$ is not an integer multiple of $s$ and $x + 1 < (\rhalf + 1) s$.
  \end{itemize}
    \end{lemma}
    
    \begin{proof} Observe first that $r_0 \le r_1 \le r_0 + 2$.  The first inequality is obvious (add a light good to the lighter bundle) and the second inequality holds because we may remove two heavy goods from the heavier bundle and add $2s - 1$ light goods. Also $r_0$ and $r_1$ have the same parity. Finally $r_{\sfrac{1}{2}} - r_0 = \pm 1$ since the two numbers have different parity and we can switch between the two values by exchanging a heavy good by either $\floor{s}$ or $\ceil{s}$ light goods.

 Let $x + \sfrac{1}{2} = \rhalf s + y$ with $y \in \NN_0$. Then $x + 1 = \rhalf s + y + \sfrac{1}{2}$. If $y + \sfrac{1}{2} = s$, $r_1 = \rhalf +1$. Also $r_0 < r_1$ and hence $r_0 = \rhalf-1$. If $y + \sfrac{1}{2} > s$ then $y - \sfrac{1}{2} \ge s$ and therefore $r_1 = r_0 = \rhalf + 1$. If $y + \sfrac{1}{2} < s$, then $x + 1 = (\rhalf - 1)s + (s + y + \sfrac{1}{2})$ and $x - 1 = (\rhalf - 1)s + (s + y - \sfrac{1}{2})$ and therefore $r_1 = r_0 = \rhalf - 1$. 
\end{proof}

The following Lemma is useful for argueing that a change of allocation is feasible. 

\begin{lemma}[Global Accounting of Light Goods]\label{global accounting} Let $S$ be the total value of all goods, for each agent $i$ let $t_i$ be a target value for the agent such that $S = \sum_i t_i$, and let $D^H$ be an allocation of the heavy goods such that $t_i - \abs{D^H_i} s$ is a non-negative integer. Then the light goods can be added to the bundles so as to give each bundle its target value and all light goods are allocated. \end{lemma}
\begin{proof} Let $h_i = \abs{D^H_i}$. Then the number of available light goods is $S - \sum_i h_i s$. We need $t_i - h_is$ light goods for agent $i$ and hence $\sum_i (t_i - h_i s) = S - \sum_i h_i s$ altogether. \end{proof}

\begin{lemma}\label{simple observation} Assume $\hdeg(i) < \hdeg(j)$ and $w_i \not= w_j$. Let $\ell$ be the total number of light goods in the bundles of $i$ and $j$. Moving a heavy good from $j$ to $i$ and then adding the light goods greedily does not decrease the $\NSW$. \end{lemma}
\begin{proof} The absolute difference in heavy weight decreases except if $\hdeg(i) + 1 = \hdeg(j)$ when it is unchanged. Let $d$ be the difference in heavy weight after the reallocation of the heavy good. If $d > \ell$, the greedy allocation reduces the weight difference to $d - \ell$ and the difference before the reallocation was at least as large. If $d \le \ell$ and $d$ is half-integral, the greedy allocation reduces the weight difference to $\sfrac{1}{2}$ and the weight difference was at least as large before. If $d$ is integral, the greedy allocation reduces the weight difference to either zero or one. The weight difference before was at least one since $w_i \not= w_j$ and the difference in heavy weight was integral before the reallocation. \end{proof}

\newcommand{\tw}{\mathit{tw}}

We use $A_d$ and $O_d$ for the set of agents owning bundles of value $x + d$ in $A$ and $O$, respectively. 

  \begin{lemma}[Lemmas 14 and 15 in~\cite{NSW-twovalues-halfinteger}]\label{xo}
    $x - \sfrac{1}{2} \le x_O \le x + \sfrac{1}{2}$. The bundles in $O$ have values in $\sset{x - \sfrac{1}{2}, x, x + \sfrac{1}{2}, x + 1}$. 
\end{lemma}
\begin{proof} The average value of a bundle in $O$ is less than $x + 1$. Therefore $x_O \le x + \sfrac{1}{2}$. 
Assume $x_O \le x - 1$.  In $O$, all light goods are contained in bundles of value at most $x_O + 1$, and hence bundles of value larger than $x$ are heavy-only. Any bundle contains at most $k_0 - 1$ heavy goods (Lemma~\ref{agree on heavy}) and hence has heavy value at most $(k_0 - 1)s$.
Let $y = (k_0 - 1)s$. If $y \le x$, the average value of a bundle in $O$ is strictly less than $x$ (there is a bundle of value $x_O$ and all bundles have value at most $x$), but the average value of a bundle in $A$ is at least $x$, a contradiction. So $y \in \{x + \sfrac{1}{2}, x + 1\}$. In $O$, bundles of value $y$ are heavy-only. Let $U$ be the set of owners of the bundles of value $y$ in $O$. If their bundles in $A$ have value $y$ or more, the average value of a bundle in $A$ is larger than the average value in $O$, a contradiction. Note that bundles in $N \setminus U$ have value at least $x$ in $A$, have value at most $x$ in $O$, and there is a bundle of value $x_O$ in $O$.

  So there is an agent $i \in U$ whose bundle in $A$ has value less than $y$. Then $\hdeg_A(i) \le k_0 - 2 < \hdeg_O(i)$. Consider an alternating path starting in $i$ starting with an $O$-edge. It ends in a node $j$ with $\hdeg_O(j) < \hdeg_A(j) \le k_0 - 1$. Then $\hdeg_O(j) <\hdeg_O(i)$ and the value of $O_j$ is at most $x$ and hence less than $y$. Augmenting the path to $O$ effectively moves a heavy good from $i$ to $j$ and decreases the distance between $A^H$ and $\OH$. Adding the light goods in $O_i \cup O_j$ greedily will result in an allocation whose $\NSW$ is at least the one of $O$ (Lemma~\ref{simple observation}). Contradiction.

It remains to show that bundles in $O$ have value at most $x + 1$. Assume there is a bundle of value $x + \sfrac{3}{2}$ or more in $O$. The bundle
  contains at most $k_0 - 1$ heavy goods and hence its heavy value is at most $x + 1$. So it contains a light good and hence $x_O = x + \sfrac{1}{2}$, and the bundle under consideration has value $x + \sfrac{3}{2}$.
The bundles in $O$ have values in $\{x + \sfrac{1}{2}, x + 1, x + \sfrac{3}{2}\}$. Any bundle of value $x + \sfrac{1}{2}$ can be turned into a bundle of value $x +\sfrac{3}{2}$ by moving a light good to it from a bundle of value $x + \sfrac{3}{2}$.

If $A$ is optimal, $O^H = A^H$ as $O$ is closest to $A$. Hence $x_O = x_A = x$, and there are no bundles of value $x + \sfrac{3}{2}$ in $O$. If $A$ is sub-optimal, there are more bundles of value $x + \sfrac{1}{2}$ in $O$ than in $A$. So there is an $i$ in $(O_{\sfrac{1}{2}} \cup O_{\sfrac{3}{2}}) \cap (A_0 \cup A_1)$. Then the parities of $\hdeg_A(i)$ and $\hdeg_O(i)$ are different.  We may assume $i \in O_{\sfrac{3}{2}}$.

Assume first that $\hdeg_O(i) > \hdeg_A(i)$. Then there exists an alternating path starting in $i$ with an $O$-edge. The path ends in $j$ with $\hdeg_O(j) < \hdeg_A(j) \le k_0 - 1$. The heavy value of $O_j$ is at most  $(k_0 - 2)s$ which is at most $x + 1 - s$. Since the value of $O_j$ is at least $x_O$, $O_j$ contains at least $x_O - (x + 1 - s) = s - \sfrac{1}{2} = \floor{s}$ light goods. 
If $w_j^O = x + \sfrac{1}{2}$, we augment the path, move $\floor{s}$ light goods  from $j$ to $i$ and improve $O$, a contradiction. If $w_j^O = x + 1$, we augment the path and move $\floor{s}$ light goods from $j$ to $i$. This does not change the utility profile of $O$ and moves $O$ closer to $A$, a contradiction. If $w_j^O = x + \sfrac{3}{2}$, $j$ contains at least $\ceil{s}$ light goods. We augment the path, move $\floor{s}$ light goods from $j$ to $i$ and one light good from $j$ to a bundle of value $x_O$. This improves $O$, a contradiction. 

Assume next that $\hdeg_O(i) < \hdeg_A(i) \le k_0 - 1$. Then there exists an alternating path starting in $i$ with an $A$-edge. The path ends in $j$ with $\hdeg_A(j) < \hdeg_O(j) \le k_0 - 1$. Since $O_i$ has value $x + \sfrac{3}{2}$ and $\hdeg_O(i) \le k_0 - 2$, $O_i$ contains at least $x + \sfrac{3}{2} - (x  + 1 - s) = \ceil{s}$ light goods. We augment the path to $O$ and remove $\ceil{s}$ light goods from $O_i$. So the value of $O_i$ becomes $x + 1$. If $O_j$ has value $x + \sfrac{3}{2}$, we put $\floor{s}$ light goods on $j$ and one light good on a bundle of value $x_O$. If $O_j$ has value $x + 1$ or $x + \sfrac{1}{2}$, we put $\ceil{s}$ light goods on $j$. This either improves $O$ or does not change the utility profile of $O$ and moves $O$ closer to $A$, a contradiction. 
\end{proof}

\begin{lemma}\label{containsalight}\label{$x + 1$ contains a light}\label{heavy-only $x + 1$} 
Consider a heavy-only agent $i \in A_1$. 
Then $O_i$ is heavy-only and has value $x + 1$. If all bundles in $A_1$ are heavy-only, $A$ is  optimal.
\end{lemma}
\begin{proof} Consider a heavy-only bundle $A_i$ of value $x + 1$. Then $\hdeg_A(i) = k_0 - 1$ and $x + 1 = (k_0 - 1)s$. Assume for the sake of a contradiction that $O_i$ has either value less than $x +1$ or is not heavy-only. In either case, $\hdeg_A(i) > \hdeg_O(i)$ and hence $O_i$ contains at most $k_0 - 2$ heavy goods and thus at least $w_i^O - (x + 1 - s) = w_i^O - x - 1 + s$ light goods. If $i \in O_0 \cup O_1$, $O_i$ contains at most $k_0 - 3$ heavy goods as the parity of the number of heavy goods is the same as for $A_i$. 

  Consider an alternating path starting in $i$ with an $A$-edge and let $j$ be the other end of the path. Then $\hdeg_A(j) < \hdeg_O(j)$ and hence $A_j$ contains at most $k_0 - 2$ heavy goods; its heavy value is therefore at most $x + 1 - s$. Since the value of $A_j$ is at least $x$, $A_j$ contains  at least $\floor{s}$ light goods. At least $\ceil{s}$, if $A_j$ has value $x + 1$.

  If $w_j^A \in \sset{x,x+1}$, we augment the path to $A$, move $\floor{s}$ light goods from $j$ to $i$, and if $w_j^A = x + 1$, in addition, move one light good from $j$ to a bundle of value $x$. In this way, we increase the number of bundles in $\Ahalf$ by two, a contradiction to property (3). 

  If $w_j^A = x + \sfrac{1}{2}$ and $A_j$ contains at least $\ceil{s}$ light goods\footnote{It will contain at least $2s + \floor{s}$ light goods}, we augment the path to $A$, and move $\floor{s}$ light goods from $j$ to $i$. In this way, $i$ and $j$ interchange values and $A_j$ becomes a bundle of value $x + 1$ containing a light good, a contradiction to property (3).

If $w_j^A = x + \sfrac{1}{2}$ and  $A_j$ contains exactly $\floor{s}$ light goods, it contains exactly $k_0 - 2$ heavy goods. Then $O_j$ contains $k_0 - 1$ heavy goods and hence is a heavy-only bundle of value $x + 1$. If the value of $O_i$ is $x - \sfrac{1}{2}$, $O_i$ contains at least $s - \sfrac{3}{2}$ light goods. We augment the path to $O$ and move $s - \sfrac{3}{2}$ light goods from $i$ to $j$. This does not change the utility profile of $O$ and moves $O$ closer to $A$, a contradiction. If the value of $O_i$ is either $x$ or $x + \sfrac{1}{2}$, $O_i$ contains at least $\floor{s}$ light goods. We augment the path to $O$ and move $\floor{s}$ light goods from $i$ to $j$. This improves the $\NSW$ of $O$ if the value of $O_i$ is $x$ and does not change the utility profile of $O$ and moves $O$ closer to $A$, otherwise. If the value of $O_i$ is $x + 1$, $O_i$ contains at least $\ceil{s}$ light goods. We augment the path to $O$, move $\floor{s}$ light goods  from $i$ to $j$ and one light good from $i$ to a bundle of value $x$. This improves $O$. In either case, we have obtained a contradiction. 

If all bundles in $A_1$ are heavy-only, $A_1 \subseteq O_1$ and hence $A$ is optimal by Lemma~\ref{less-more}.
 \end{proof}

\newcommand{\Opm}{O_{\pm \sfrac{1}{2}}}

  We next decompose $A^H \oplus \OH$ into \emph{walk}s. Walks were previously used in studies of parity matchings~\cite{Akiyama-Kano}; see Section~\ref{Algorithmics}. For an edge in $A^H \setminus \OH$, we also say that the edge has type $A$. Similarly, edges in $\OH \setminus A^H$ have type $O$. We use $T \in \sset{A,O}$ for the generic type. Then $\bar{T}$ is the other type. For $i \in \Ahalf \cap \Opm$ the heavy parity is the same in both allocations; we call the node \emph{unbalanced} if the cardinalities of $A^H_i$ and $\OH_i$ are different. We pair the edges in $A^H_i \oplus \OH_i$, first edges of different types and then the remaining either $A$- or $O$-edges (there is an even number of them). Goods have degree zero or two in $\OH \oplus \AH$. For goods of degree two, we pair the incident $A$- and $O$- edge. In this way, we form walks. The endpoints of these walks are agents in $A_0 \cup A_1 \cup O_0 \cup O_1$. The interior nodes of the walks lie in $\Ahalf \cap \Opm$. We distinguish two kinds of interior nodes: through-nodes and hinges. In through-nodes the two incident edges have different types, in hinges the incident edges have the same type. Hinges are unbalanced. Depending on the type, the hinge is an $O$- or $A$-hinge. The types of the hinges along a walk alternate. Endpoints are also either $A$- or $O$- endpoints depending on the type of the incident edge. The paths between hinges are alternating. Imagine, we follow the walk from one endpoint to the other. If the first endpoint has type $T$, the first hinge will have type $\bar{T}$, the second hinge will have type $T$, the last hinge and second endpoint will have type $\bar{T}$ and $T$ or $T$ and $\bar{T}$ depending on whether the number of hinges is odd or even. 

Since there are more agents of value $x + \sfrac{1}{2}$ in $O$ than in $A$, there will be a walk $W$ having an endpoint $i \in \Ohalf \cap (A_0 \cup A_1)$. Let $j$ be the other endpoint. Then $j \in A_0 \cup A_1 \cup O_0 \cup O_1$. We will show: If $j \in A_0 \cup A_1$, $A$ does not satisfy its defining properties, and if $j \in (O_0 \cup O_1) \cap \Ahalf$, $O$ is not closest to $A$. .
Hinge nodes lie in $A_{\sfrac{1}{2}} \cap O_{\pm \sfrac{1}{2}}$. Hinge nodes actually lie in $O_{\sfrac{1}{2}}$ and $T$-hinges own at least $2s$ light goods in the allocation $\bar{T}$ for $T \in \{A,O\}$ as we show next.

\begin{lemma}[Lemma 23 in~\cite{NSW-twovalues-halfinteger}]\label{lights in hinges} All hinge nodes of $W$ belong to $O_{\sfrac{1}{2}}$, $A$-hinges own at least $2s$ light goods in $O$ and $O$-hinges own at least $2s$ light goods in $A$. \end{lemma}
\begin{proof} By  definition, hinge nodes are unbalanced and lie in
  $A_{\sfrac{1}{2}} \cap O_{\pm \sfrac{1}{2}}$. Consider two consecutive hinges $h$ and $h'$ and the alternating path $P$ connecting them. Assume that their values in $O$ differ by one, i.e, one has value $x + \sfrac{1}{2}$ and the other value $x - \sfrac{1}{2}$. We change $O^H$ to $O^H \oplus P$ and use Lemma~\ref{global accounting} to show that the modified heavy allocation can be extended to a full allocation. 
We define the target value of both endpoints as $x$ and the target value of all other agents as their value in $A$. Then the sum of the values is unchanged. The $A$-endpoint of $P$ receives an additional heavy good; this is fine as $\hdeg_O(h) < r_{\sfrac{1}{2}}$ and hence $\hdeg_O(h) \le r_{\sfrac{1}{2}} - 2$ and therefore $\hdeg_O(h) + 1 \le r_{\sfrac{1}{2}} + 1 \le r_0$ according to Lemma~\ref{max number of heavy}. The $O$-endpoint loses a heavy good; this is fine since $\hdeg_O(h) > 0$. The heavy degree of all other nodes stays unchanged. The modified allocation has smaller $\NSW$ than $O$, a contradiction. We have now shown that all hinge nodes have the same value in $O$.

It remains to show that the first hinge of the walk lies in $O_{\sfrac{1}{2}}$; call it $h$. Assume $h \in O_{-\sfrac{1}{2}}$. We distinguish cases according to whether $i$ is $A$-heavy or not. We define the target value of $i$ and $h$ as $x$. Then again the sum of the values does not change. If an allocation with the desired target values exists, it has larger $\NSW$ than $O$, a contradiction. 

If $i$ is $O$-heavy, $h$ is $A$-heavy. We augment $P$ to $O$. Then $i$ looses a heavy good (this is fine as $\hdeg_O(i) > 0$) and $h$ gains a heavy good; this is fine as $h$ is an $A$-hinge and therefore $\hdeg_O(i) \le \hdeg_A(i) - 2 \le r_{\sfrac{1}{2}} - 2$ and hence  $\hdeg_O(i) + 1 \le r_{\sfrac{1}{2}} - 1 \le r_0$. 

If $i$ is $A$-heavy, $h$ is $O$-heavy. We augment $P$ to $O$. Then $h$ looses a heavy good; this is fine as $\hdeg_O(h) > 0$. The agent $i$ gains a heavy good; this is fine as $i$ is an $A$-endpoint and therefore $\hdeg_O(i) \le \hdeg_A(i) - 2 \le r_{\sfrac{1}{2}} - 2$ and hence  $\hdeg_O(i) + 1 \le r_{\sfrac{1}{2}} - 1 \le r_0$. 

We have now established that all hinges are unbalanced and belong to $A_{\sfrac{1}{2}} \cap O_{\sfrac{1}{2}}$. For a $T$-hinge $h$, $\hdeg_{\bar{T}}(v) \le \hdeg_T(v) - 2$. Thus $h$ contains at least $2s$ light goods in $\bar{T}$.\end{proof}

  \begin{lemma} Any walk $W$ in the decomposition of $A^H \oplus \OH$ is a simple path, i.e., no agent occurs twice on $W$. \end{lemma}
  \begin{proof} Assume otherwise, say $W$ visits a vertex $i$ twice. Consider the path $W'$ between the two occurences of $i$. The first and the last edge of $W'$ have either different types or the same type. If they have the same type, $W'$ contains an odd number of hinges and hence together with $i$ forms a cycle with an even number of hinges. If the first edge and the last edge have different types, $W'$ contains an even number of hinges, maybe zero. If there is no hinge, $W'$ is simply an alternating cycle. If there is a hinge, we combine the alternating path from $i$ to the first hinge on $W'$ and the alternating path from the last hinge of $W'$ to $i$ into a single alternating path. So in either case, $W'$ is a cycle with an even number of hinges that are connected by alternating paths. Half of the hinges are $A$-hinges and half are $O$-hinges. Let us pair the hinges so that each pair contains an $O$-hinge and an $A$-hinge. We augment $W'$ to $O$. Through-nodes gain and loose a heavy good and hence the value of their bundle does not change. $O$-hinges loose two heavy edges and $A$-hinges gain two heavy edges.
Consider a pair $(h,\ell)$ of $A$- and $O$-hinge. An $A$-hinge owns at least $2s$ light goods in $O$. We give these goods to $\ell$. Then the utility profile of $O$ does not change and $O$ moves closer to $A$, a contradiction. \end{proof}

At this point, we have established the existence of a walk $W$ with endpoint $i \in (A_0 \cup A_1) \cap O_{\sfrac{1}{2}}$. All hinge nodes belong to $A_{1/2} \cap O_{1/2}$ and $T$-hinges own $2s$  light goods in the allocation $\bar{T}$. We will next show that we can use $W$ to either improve $A$ or derive a contradiction to the assumption that $O$ is closest to $A$. We may assume that there is a bundle of value $x + 1$ in $A$ containing a light good. Otherwise, $A$ is optimal (Lemma~\ref{containsalight}) and we are done. 
Let $j$ be the other endpoint of the walk. Since $W$ is simple, $j \not= i$.

Assume $i$ is $T$-heavy, i.e. $\hdeg_T(i) > \hdeg_{\bar{T}}(i)$. The value of $A_i$ is $x$ or $x + 1$ and the value of $O_i$ is $x + \sfrac{1}{2}$. The types of the hinges alternate along $W$, the type of the first (last) hinge is opposite to the type of $i$ ($j$). Each $A$-hinge holds $2s$ light goods in $O$ and each $O$-hinge holds $2s$ light goods in $A$ (Lemma~\ref{lights in hinges}). If the types of $i$ and $j$ differ, the number of hinges is even, if the types are the same, the number of hinges is odd. There is also a bundle of value $x + 1$ in $A$ containing a light good and there is a bundle of value $x$ in $A$. 

We distinguish two cases: $j \in A_0 \cup A_1$ and $j \in (O_0 \cup O_1) \cap A_{\sfrac{1}{2}}$. In the first case, we show how to improve $A$ and in the other case we will derive a contradiction to the assumption that $O$ is closest to $A$.

\paragraph{Case $j \in A_0 \cup A_1$:} We augment the walk to $A$. The heavy parity of $i$ and $j$ changes and the heavy parity of all nodes, including the hinges, does not change. Each $O$-hinge releases $2s$ light goods and each $A$-hinge requires $2s$ light goods. An $A$-endpoint looses a heavy good and an $O$-endpoint gains a heavy good.

We define the target values of $i$ and $j$ as $\tw_i = \tw_j = x + \sfrac{1}{2}$. If $w^A_i = w^A_j = x$, let $k$ be an agent owning a bundle of value $x + 1$ containing a light good and define $\tw_k = x$. If $w^A_i = w^A_j = x + 1$, let $k$ be an agent owning a bundle of value $x$ and define $\tw_k = x + 1$. If $\sset{w^A_i,w^A_j} = \sset{x, x + 1}$, leave $k$ undefined. For all other agents, their target value is their value in $A$. Then the sum of the target values is the same as the total value of the bundles in $A$.

We next define the number of heavy goods allocated to each agent in the new allocation. For an $A$-hinge $h$ the number decreases by two; since $\hdeg_A(h) \ge \hdeg_0(h) + 2 \ge 2$, this is fine. For an $O$-hinge $h$ the number increases by two (since $\hdeg_A(h) + 2 \le \hdeg_O(h) \le r_{\sfrac{1}{2}}$ this is fine), for an $O$-endpoint the number increases by one; if $i$ is an $O$-endpoint, this is fine since $\hdeg_A(i) < \hdeg_O(i) \le r_{\sfrac{1}{2}}$; if $j$ is an $O$-endpoint and $j \in \Ohalf$, the same reasoning applies; if $j$ is an $O$-endpoint and $j \in O_0 \cup O_1$, $\hdeg_A(j) < \hdeg_O(j)$ implies $\hdeg_A(j) + 2 \le \hdeg_O(j) \le r_1$ and hence $\hdeg_A(j) + 1 \le r_{\sfrac{1}{2}}$. For an $A$-endpoint $h$ the number decreases by one; this is fine, since $\hdeg_A(h) > \hdeg_O(h) \ge 0$. Finally, the total number of allocated heavy goods stays the same. This holds since the number of hinges is even if the endpoints have differnt types, and since there is an additional $T$-hinge if both endpoints have type $\bar{T}$. 

The total number of allocated heavy goods stays the same, the total value of the bundles stays the same, and each bundle contains at most the maximum feasible number of heavy goods. Therefore, the number of light goods is precisely the number of light goods needed to fill all bundles to their target value (Lemma~\ref{global accounting}). The transformation improves the $\NSW$ of $A$ and the values of all bundles in the new allocation lie in $\sset{x, x+ \sfrac{1}{2}, x + 1}$, a contradiction to the optimality of $A$ under these constraints.

\paragraph{Case $j \in (O_0 \cup O_1) \cap A_{\sfrac{1}{2}}$:} We augment the walk to $O$. For the intermediate nodes, the heavy parity does not change. For $i$ and $j$ the heavy parity changes. We define the target values as $\tw_i = w^O_j$ and $\tw_j = w^O_i$. For all other agents the target value is their value in $O$. So the sum of the target values is the total value of the bundles in $O$ and the utility profile does not change.

If $i$ and $j$ have different types, the number of hinges is even. If $i$ and $j$ have type $T$, there is an extra hinge of type $\bar{T}$. If $T = A$, the endpoints gain one heavy good each and the hinge looses two heavy goods. If $T = O$, the endpoints loose one heavy good each and the hinge gains two heavy goods. So the sum of the heavy goods in all bundles does not change.

It remains to argue that no bundle contains too many heavy goods. A $O$-endpoint $h$ looses a heavy good; this is fine as $\hdeg_O(h) > 0$ for an $O$-endpoint. An $A$-endpoint gains a heavy good. If $j$ is an $A$-endpoint, $\hdeg_O(j) < \hdeg_A(j) \le r_{\sfrac{1}{2}}$ and hence $\hdeg_O(j) + 1 \le r_{\sfrac{1}{2}}$; this is fine, since $\tw_j = w^O_i = x + \sfrac{1}{2}$. If $i$ is an $A$-endpoint and $A_i$ has value $x$, $\hdeg_O(i) < \hdeg_A(i) \le r_0$ and hence $\hdeg_O(i) + 1 \le r_0$. This is fine, since $\tw_i = w_j^O \in \sset{x, x + 1}$. If $i$ is an $A$-endpoint and has value $x + 1$, $A_i$ cannot be heavy-only since then $O_i$ would also have value $x + 1$ and be heavy-only according to Lemma~\ref{heavy-only $x + 1$}. Hence $\hdeg_O(i) \le \hdeg_A(i) - 2 \le r_1 - 2$. So $\hdeg_O (i) + 1 \le r_1 -1 \le r_O$, and we are fine. 

Thus the utility profile of $O$ does not change and $O$ moves closer to $A$, a contradiction.

\begin{theorem} $A$ is optimal.\end{theorem}
\begin{proof} Assume $O$ is closest to $A$ and $A$ is sub-optimal. Then there is an improving walk $W$. Augmentation of $W$ to $A$ moves two agents from $A_0 \cup A_1$ to $A_{\sfrac{1}{2}}$ and does not change the value of any other bundle, a contradiction to the assumption that the number of bundles of value $x + \sfrac{1}{2}$ in $A$ is maximum.
\end{proof}

  \section{The Algorithmics of Optimizing the Small Bundles}\label{Algorithmics}

  Let $x$ be the minimum value of a bundle after the greedy addition of the light goods. Then only bundles of value $x$, $x + \sfrac{1}{2}$ and $x + 1$ can contain small goods. Let $N_s$ be the owners of the bundles of value at most $x + 1$ and let $\Gamma(N_s)$ be the goods assigned to $N_s$ in $A$. We need to construct the allocation $C$ of the goods in $\Gamma(N_s)$ to $N_s$ maximizing the number of bundles of value $x + \sfrac{1}{2}$.
All alternating paths in this section are with respect to $A$ and $\bar{A}$.

  \begin{lemma}\label{easy case} Let $A$ be an allocation with all values in $\{x, x + \sfrac{1}{2}, x + 1\}$. If there are either no bundles of value $x$ or no bundles of value $x + 1$, $A$ is optimal. \end{lemma}
  \begin{proof} In either case, the number of bundles of each type is determined by the total value of the bundles and the number of bundles. \end{proof}

 \begin{lemma}\label{create $x + 1$ containing a light} Assume all bundles in $A_1$ are heavy-only. If $A$ is sub-optimal there is an alternating path starting with an $A$-edge from an agent in $A_1$ and ending in either $A_0$ or in an agent in $A_{\sfrac{1}{2}}$ owning more than $\floor{s}$ light goods. \end{lemma}
 \begin{proof} Since $x + 1$ is a multiple of $s$, bundles of value $x$ contain at least $2s - 1$ light goods and bundles of value $x + \sfrac{1}{2}$ contain at least $\floor{s}$ light goods. Consider alternating paths starting with an $A$-edge at agents in $A_1$. If there is a path to an agent $j$ of value $x$ augmenting the path to $A$ and moving $\floor{s}$ goods from $j$ to $i$ changes the weight of $i$ and $j$ to $x + \sfrac{1}{2}$ and improves the $\NSW$ of $A$. If there is a path to an agent $j$ of value $x + \sfrac{1}{2}$ owning more than $\floor{s}$ light goods\footnote{$j$ then owns at least $2s + \floor{s}$ light goods.} augmenting the path and moving $\floor{s}$ light goods from $j$ to $i$ interchanges the values of $i$ and $j$ and creates a bundle of value $x + 1$ containing a light good. If neither kind of path exists, all alternating paths starting with an $A$-edge from an agent in $A_1$ either lead to an agent in $A_1$  or to an agent in $A_{\sfrac{1}{2}}$ owning $\floor{s}$ light goods. Let $U$ be this set of agents. All heavy goods owned agents in $U$ are light for agents not in $U$ and $\floor{s}$ light goods must be allocated to every agent in $U \cap A_{\sfrac{1}{2}}$. The bundles owned by the agents in $\bar{U}$ have value $x$ and $x + \sfrac{1}{2}$, and the maximal number of heavy and light goods are assigned to the agents in $\bar{U}$. So the allocation is optimal. \end{proof}

For $d \in \sset{0,\sfrac{1}{2},1}$, let $N_d = \set{\ell}{\text{$\ell \le r_d$ and $\ell$ and $r_d$ have the same parity}}$ be the set of permissible numbers of heavy goods in bundles of value $x + d$.

\begin{lemma}[Lemma 27 in~\cite{NSW-twovalues-halfinteger}]\label{parity matching} Let $A$ be an allocation with all values in $\{x, x + \sfrac{1}{2}, x + 1\}$. $A$ is sub-optimal if and only if there is an allocation $C^H$ of the heavy goods in $A$ and a pair of agents $i$ and $j$ in $A_0 \cup A_1$ such that 
\begin{itemize}
    \item all agents in $A_{\sfrac{1}{2}} \cup \{i,j\}$ own a number of heavy goods in $N_{1/2}$ and for each of the agents in $A_0 \cup A_1 \setminus \{i,j\}$, the number of owned heavy goods is in the same $N$-set as in $A$, and, 
    \item if $i$ and $j$ own bundles of value $x$ in $A$, there must be a bundle of value $x + 1$ in $A$ containing a light good and if $i$ and $j$ own bundles of value $x + 1$ in $A$, there must be a bundle of value $x$ in $A$. 
    \end{itemize}
    \end{lemma}
\begin{proof} If $A$ is sub-optimal there is an improving walk $W$.  Let $i$ and $j$ be the endpoints of the walk. Augmenting the walk and moving the light goods around as described in Section~\ref{APC exists}
  \begin{itemize}
  \item adds $i$ and $j$ to $A_{1/2}$,
  \item reduces the value of a bundle of value $x + 1$ containing a light good to $x$ if $A_i$ and $A_j$ have value $x$ and increases the value of a bundle of value $x$ to $x + 1$ if $A_i$ and $A_j$ have value $x + 1$, and 
  \item leaves the value of all other bundles unchanged.
  \end{itemize}
 Thus, in the new allocation, the number of heavy goods owned by $i$ and $j$ lies in $N_{\sfrac{1}{2}}$. For all other agents the number of owned heavy goods stays in the same $N$-set. This proves the only-if direction. 

We turn to the if-direction. Assume that there is an allocation $C^H$ of the heavy goods in which for two additional agents $i$ and $j$ the number of owned heavy goods lies in $N_{1/2}$ and for all other agents the number of owned heavy goods stays in the same $N$-set. We will show how to allocate the light goods such that $C^H$ becomes an allocation $C$, in which all bundles have value in $\{x, x + \sfrac{1}{2}, x + 1\}$ and $C_{\sfrac{1}{2}} = A_{\sfrac{1}{2}} \cup \{i,j\}$. Then the $\NSW$ of $C$ is higher than the one of $A$.

We next define the values of the bundles in $C$ and in this way fix the number of light goods that are required for each bundle. For $i$ and $j$, we define $\tw^B_i = \tw^B_j = x + \sfrac{1}{2}$. If $A_i$ and $A_j$ have both value $x$, let $k$ be an agent owning a bundle of value $x +1$ containing a light good and define $w_k^C = x$. If $A_i$ and $A_j$ have both value $x + 1$, let $k$ be an agent owning a bundle of value $x$ and define $w_k^C = x + 1$.  Then $w_i^A + w_j^A + w_k^A = w_i^C + w_j^C + w_k^C$ in both cases. If one of $A_i$ and $A_j$ has value $x$ and the other one has value $x + 1$, let $k$ be undefined. Then $w_i^A + w_j^A = w_i^C + w_j^C$. For all $\ell$ different from $i$, $j$, and $k$, let $w_\ell^A = w_\ell^C$. Then the total value of the bundles in $A$ and $C$ is the same. The total heavy value is also the same. Also, for each agents the heavy value is at most the target value. So the total difference between the target values and the heavy values is exactly the same in $C$ as in $A$. Hence $C$ exists. \end{proof}

In a \emph{parity matching problem} in a bipartite graph a maximum degree is specified for every vertex. The question is whether there is a subset of the edges such that the degree of every vertex is at most its maximum and has the same parity as its maximum. Parity matching problems can be solved in polynomial time~\cite{CornuejolsFactors,Akiyama-Kano}.

\begin{lemma}\label{running time} An optimal allocation for $N_s$ can be computed in polynomial time. The $\NSW$-allocation for half-integer instances can be computed in polynomial time. 
\end{lemma}
\begin{proof} Let $A$ be an allocation for $N_s$. All bundles have values in $\sset{x, x + \sfrac{1}{2}, x + 1}$. If there are no bundles of value $x$ or no bundles of value $x + 1$, $A$ is optimal (Lemma~\ref{easy case}).

  If all bundles of value $x + 1$ are heavy-only, we search for an alternating path starting with an $A$-edge from an agent in $A_1$ and ending either in an agent in $A_0$ or in an agent in $\Ahalf$ owning more than $\floor{s}$ light goods. If such a path exists, we either improve $A$ or create a bundle of value $x + 1$ containing a light good. We repeat until we create a bundle of value $x +1$ containing a light good. If we do not succeed, $A$ is optimal (Lemma~\ref{create $x + 1$ containing a light}).

So assume now that we have a bundle of value $x$ and a bundle of value $x + 1$ containing a light. We then check whether $A$ can be improved according to Lemma~\ref{parity matching}. 
In order to check for the existence of the allocation $C^H$, we set up the following parity matching problem for every pair $i$ and $j$ of agents. 
\begin{itemize}
\item For goods the degree in the matching must be equal to 1.
\item For all agents in $A_{\sfrac{1}{2}} \cup \{i,j\}$, the degree must be in $N_{\sfrac{1}{2}}$.
\item If $A_i$ and $A_j$ have value $x$, let $A_k$ be any bundle of value $x + 1$ containing a light good. If $A_i$ and $A_j$ have value $x + 1$, let $A_k$ be a bundle of value $x$. The degree of $k$ must be in $N_0$.\footnote{It would be incorrect to require that the degree of $k$ must be in $N_1$ because we want to allocate a light good to $k$.}
\item For an $a \in A_0\setminus \{i,j,k\}$, the degree must be in $N_0$, and for an $a \in A_1 \setminus \{i,j,k\}$, the degree must be in $N_1$. 
\end{itemize}
If $C^H$ exists for some pair $i$ and $j$, we improve the allocation. If $C^H$ does not exist for any pair $i$ and $j$, $A$ is optimal.

Each improvement increases the size of $A_{\sfrac{1}{2}}$ by two and hence there can be at most $n/2$ improvements. In order to check for an improvement, we need to solve $n^2$ parity matching problems. Parity matching problems can be solved in polynomial time.
\end{proof}


\bibliographystyle{alpha}
\bibliography{ref}

\end{document}

%% file: papermacros.tex
\input bookmacros.tex

\newlength{\proofpostskipamount}\newlength{\proofpreskipamount}
\newenvironment{proof}%
               {\par\vspace{\proofpreskipamount}\noindent{\bf Proof:}\hspace{0.5em}}
               {\nopagebreak%
                \strut\nopagebreak%
                \hspace{\fill}\qed\par\vspace{\proofpostskipamount}\noindent}

               {\par\vspace{0.5ex}\noindent{\bf Proof #1:}\hspace{0.5em}}%
               {\nopagebreak%
                \strut\nopagebreak%
                \hspace{\fill}\qed\par\medskip\noindent}

\newtheorem{theorem}{Theorem}
\newtheorem{lemma}{Lemma}

\providecommand{\qed}{\rule[-0.2ex]{0.3em}{1.4ex}}

\newlength{\mydefwidth}
\newlength{\mytextwidth}

\newcommand{\myurl}[1]{{\footnotesize \url{#1}}}

%% file: bookmacros.tex
\newcommand{\floor}[1]{\left\lfloor #1 \right\rfloor}
\newcommand{\ceil}[1]{\left\lceil #1 \right\rceil}

\newcommand{\donotshow}[1]{}

\newcommand{\ignore}[1]{}

\providecommand{\CC}{C\raisebox{.08ex}{\hbox{\tt ++}}}

\newcommand {\abs}[1] {| #1 |}

\newlength{\mysetspacing}
\providecommand{\sset}[1]{\{\hspace{\mysetspacing} #1 \hspace{\mysetspacing}\}}
\providecommand{\set}[2]{ \left\{\hspace{\mysetspacing} #1 \mbox{{\rm ; }} #2 \hspace{\mysetspacing} \right\}  }

\newcommand{\mbegin}{\{\ \ }
\newcommand{\mend}{\}}

\newlength{\mleftindent}
\newlength{\mindent}
\newlength{\mboxwidth}
\newcommand{\mincrement}{\addtolength{\mboxwidth}{-\mindent}}
\newcommand{\mdecrement}{\addtolength{\mboxwidth}{\mindent}}

\newlength{\preprogramskip}
\newlength{\postprogramskip}
\newlength{\mexpwidth}
\newlength{\mexpindent}
\newcommand{\indentafterkeyword}{\hspace*{0.5em}}

\newcommand{\mslifelse}[3]  
{\setlength{\mexpwidth}{\mboxwidth}%
\settowidth{\mexpindent}{{\bf if\indentafterkeyword}}%
\addtolength{\mexpwidth}{-\mexpindent}%
{\bf if\indentafterkeyword}\parbox[t]{\mexpwidth}{#1}\\
\mincrement \mbegin \parbox[t]{\mboxwidth}{#2 \mend} \mdecrement \\
{\bf else} \\
\mincrement \mbegin \parbox[t]{\mboxwidth}{#3}\\
\mend \mdecrement
}

{\vspace{-0.3em}\begin{itemize}
\setlength{\itemsep}{0.2\itemsep}%
\setlength{\parskip}{0.2\parskip}%
\setlength{\topsep}{0.0\topsep}%
}
{\end{itemize}}

{\begin{itemize}
\setlength{\itemsep}{0.2\itemsep}%
\setlength{\parskip}{0.2\parskip}%
\setlength{\topsep}{0.0\topsep}%
}
{\end{itemize}}

{\begin{itemize}
\setlength{\itemsep}{1.5\itemsep}%
\setlength{\parskip}{1.5\parskip}%
\setlength{\topsep}{0.0\topsep}%
}
{\end{itemize}}

{\begin{itemize}
\setlength{\itemsep}{1.5\itemsep}%
\setlength{\parskip}{1.5\parskip}%
\par \smallskip
}
{\end{itemize}}

{\begin{itemize}
\setlength{\itemsep}{1.5\itemsep}%
\setlength{\parskip}{1.5\parskip}%
\par \medskip
}
{\end{itemize}}

{\vspace{-0.3em}\begin{description}
\setlength{\itemsep}{0.2\itemsep}%
\setlength{\parskip}{0.2\parskip}%
\setlength{\topsep}{0.0\topsep}%
}
{\end{description}}

{\begin{description}
\setlength{\itemsep}{0.2\itemsep}%
\setlength{\parskip}{0.2\parskip}%
\setlength{\topsep}{0.0\topsep}%
}
{\end{description}}

{\vspace{-0.3em}\begin{enumerate}
\setlength{\itemsep}{0.2\itemsep}%
\setlength{\parskip}{0.2\parskip}%
\setlength{\topsep}{0.0\topsep}%
}
{\end{enumerate}}

{\begin{enumerate}
\setlength{\itemsep}{0.2\itemsep}%
\setlength{\parskip}{0.2\parskip}%
\setlength{\topsep}{0.0\topsep}%
}
{\end{enumerate}}